\documentclass[11pt]{article}
\usepackage{mydef2col}
\usepackage{kantlipsum,widetext}
\usepackage[T1]{fontenc}
\usepackage[autostyle=true]{csquotes}

\allowdisplaybreaks[3]
\graphicspath{{./Graphs/}}
\usepackage{tikz}
\usetikzlibrary{calc,decorations.markings}
\usetikzlibrary{arrows,patterns,positioning}

\usepackage[round]{natbib}

\def\calR{{\cal R}}
\def\calI{{\cal I}}

\def\Ito{It\^{o}'s }

\def\baru{{\bar{u}}}
\def\bth{{\bar{\theta}}}
\def\tth{{\tilde{\theta}}}
\def\br{{\bar{r}}}
\def\ba{{\bar{a}}}

\newcommand{\iu}{\mathrm{i}\mkern1mu}

\newcommand{\Res}{\operatorname{Res}}

\title{From the Black-Karasinski to the Verhulst model to accommodate the unconventional Fed's policy.  }
\def\thetitle1{From the Black-Karasinski to the Verhulst model ...}
\author{
\authorstyle{
Andrey Itkin{}
\textsuperscript{1}
Alexander Lipton{}
\textsuperscript{2}
and Dmitry Muravey
\textsuperscript{3}
}
\newline\newline
\textsuperscript{1}
\institution{Tandon School of Engineering, New York University, 1 Metro Tech Center, 10th floor, Brooklyn NY 11201, USA} \\
\textsuperscript{2}
\institution{The Jerusalem School of Business Administration, The Hebrew University of Jerusalem, Jerusalem, Israel;} \\
\textsuperscript{\ \ }
\institution{Connection Science and Engineering, Massachusetts Institute of Technology, Cambridge, MA, USA} \\
\textsuperscript{3}
\institution{Moscow State University, Moscow, Russia}
}

\date{\today}
\begin{document}

\maketitle
\thispagestyle{firstpage}

\lettrineabstract{In this paper, we argue that some of the most popular short-term interest models have to be revisited and modified to reflect current market conditions better. In particular, we propose a modification of the popular Black-Karasinski model, which is widely used by practitioners for modeling interest rates, credit, and commodities. Our adjustment gives rise to the stochastic Verhulst model, which is well-known in the population dynamics and epidemiology as a logistic model. We demonstrate that the Verhulst model's dynamics are well suited to the current economic environment and the Fed's actions. Besides, we derive an integral equations for the zero-coupon bond prices for both the BK and Verhulst (MBK) models. For the  BK model for small maturities up to 2 years, we solve the corresponding integral equation by using the reduced differential transform method. For the Verhulst PDE, under some mild assumptions, we find the closed-form solution.  Numerical examples show that computationally our approach is more efficient than the standard finite difference method.
}

\vspace{0.5in}

\section*{Introduction} \label{layer}

The Global Economic Crisis (GEC) of 2008-2010 caused unprecedented changes in the way central banks in general, and the mighty Fed in particular, conduct their business. The Quantitative Easing (QE) resulted in central banks embracing the fractional reserve \textit{modus operandi}. At the same time, commercial banks switched to the narrow bank model, partly by choice and partly by necessity.  The Federal Reserve has used short-term interest rates as the policy tool for achieving its macroeconomic goals. As a result, short rates were close to zero for much of the past decade, reflecting the effects of QE, low inflation caused by an aging population, and low productivity growth; see \citep{FRBSF2018}. The current economic recession due to the COVID-19 pandemic forces the Fed to push the short interest rates into extremely low or outright negative territory. Given the unprecedented level of unemployment, the economic recession is likely to pave the way for further use of the Fed's unconventional monetary policy, resulting in meager short rates.

Typically, short-rate interest models use stochastic drivers governed by an Ornstein-Uhlenbeck (OU) process and transform these drivers into the actual rates via suitable mappings. For instance, the Vasicek-Hull-White model uses a linear mapping, while the Black-Karasinski (BK) uses an exponential mapping. The Cox-Ingersoll-Ross model is an exception, which uses a driver governed by a Feller process.
For short-rate models driven by OU processes, the rate spends an approximately equal amount of time below and above its equilibrium level. This assumption was valid for decades. However, as was mentioned earlier, it is no longer adequate due to the nontraditional interventions of central banks. Once the rate becomes low, it tends to stay low for a very long time. Under these circumstances,  we have to revisit short-rate interest models and modify them to reflect prevailing current market conditions better. With this motivation in mind, we consider the popular BK model, which is widely used by practitioners for modeling interest rates, credit. A similar model, known as the  Schwartz one-factor model, is often used to model commodities. The enduring popularity of the model is because, despite some lack of tractability, it is relatively simple and guarantees non-negativity of rates. Besides, one can calibrate it to a given term-structure of interest rates and prices or implied volatilities of caps, floors, or European swaptions, provided that the mean-reversion level and volatility are functions of time.

For the reader's convenience, we provide some stylized facts about the BK model in Appendix~\ref{BK}. As can be seen from \eqref{BK}, the short-interest rate $r_t$ in this model is lognormal and positive. (If necessary, it can be made negative by using a deterministic shift $s(t)$.) Initially, this positivity was one of the significant advantages of the BK model. However, in the current environment, this feature seems to be less useful. Another problem with the model is in the lognormality of $r_t$. Indeed, the lognormality means the CDF of the distribution is right-skewed. Therefore, the time a typical path stays in the lower rate area is short, because the short-rate quickly moves to the mean-reversion level. Accordingly, we need to choose a low mean-reversion speed to rectify this behavior, but the qualitative behavior of $r_t$ remains the same. In contrast, given the above discussion, we should design an interest rate model with fat tails at the lower end.

In addition to the structural drawbacks, the BK model is not sufficiently tractable, especially when its coefficients are time-dependent. For instance, prices of zero-coupon bonds (ZCB) and highly liquid barrier options are not known in the closed form. We have to find these prices numerically by solving the corresponding partial differential equations (PDEs), see \eqref{VPDE}, either via finite differences or asymptotically.  In this paper, we present an attractive alternative, by deriving an integral equation for the ZCB price; see Appendix~\ref{App2}. The corresponding integral equation can also be solved numerically. Moreover, for small maturities (up to 2 years or so), it can be solved by using the reduced differential transform method; see Appendix~\ref{RDTM}. Numerical examples convincingly show that in this case, our method is more efficient computationally than the standard approach of solving PDEs via finite differences.

In this paper, we propose a modification of the BK model, which organically resolves the lower end fat tail issue, and improves the model tractability.  We describe the model in the next Section. We also present an integral equation for the ZCB price for our model and provide a closed-form solution of the model PDE in a particular case. We demonstrate that this solution accelerates the computation of the ZCB prices and provides a basis for efficient calibration.

\section{The modified BK (MBK/Verhulst) model}

Since the BK model doesn't support fat tails at the lower end, besides not being analytically tractable, we introduce its modified version of the form
\begin{align} \label{BK1}
d z_t &= \kappa(t)[\bar{\theta}(t) - e^{z_t}] dt + \sigma(t)dW_t, \\
r_t  &=  s(t) + R e^{z_t}, \qquad R = r_0 - s(0), \qquad (t,z_t) \in [0,T]\times(-\infty,\infty). \nonumber
\end{align}
In other words, we modify the dynamics of the stochastic variable $z_t$ in \eqref{BK} in the mean-reversion term by replacing $z_t$ with $e^{z_t}$.

In \eqref{BK1} $r_t$ is the short interest rate, $t$ is the time, $W_t$ is the standard Brownian motion, $\kappa(t) > 0$ is the speed of mean-reversion, $\bar{\theta}(t)$ is the mean-reversion level, $\sigma(t)$ is the volatility, $R$ is some constant  with the same dimensionality as $r_t$, eg., it can be $R = r(0) - s(0)$, T is the maturity. This model is similar to the Hull-White model, but preserves positivity of $r_t$ by exponentiating the OU random variable $z_t$. Because of that, usually practitioners add a deterministic function (shift) $s(t)$ to the definition of $r_t$ to address possible negative rates and be more flexible when calibrating the term-structure of the interest rates.

It can be seen, that at small $t$ $|z_t| \ll 1$, and so choosing $\bar{\theta}(t) = 1 + \theta(t)$ replicates the BK model in the linear approximation on $z_t$. Similarly, the choice $\bar{\theta}(t) = e^{\theta(t)}$ replicates the BK model at $z_t$ close the mean-reversion level $\theta(t)$. Thus, the modified BK model acquires the properties of the BK model while is a bit more tractable as this will be seen below.

By \Ito lemma and the Feynman–Kac formula any contingent claim written on the $r_t$ as the underlying (for instance, the price $F(\br,t,T)$ of a Zero-coupon bond (ZCB) with maturity $T$) solves the following partial differential equation
\begin{align} \label{PDEBK}
0 &= \fp{F}{t} + \dfrac{1}{2}\sigma^2(t) \br^2 \sop{F}{\br} + \kappa(t) \br [\tth(t) - \br] \fp{F}{\br} - (s(t) + R \br)F, \\
\br_t &= \frac{r_t - s(t)}{r_0 - s(0)} = e^{z_t}, \qquad \tth(t) = \bth(t) + \frac{\sigma^2(t)}{2\kappa(t)}. \nonumber
\end{align}
This equation should be solved subject to the same terminal and boundary conditions as in \eqref{termZCB1}
\begin{equation} \label{termZCB}
 F(T,\br)  = 1, \qquad F(t,\br)\Big|_{\br \to \infty} = 0.
\end{equation}
Note that since $\br \in (0,\infty)$, i.e., the boundary $\br = 0$ is not attainable, \eqref{PDEBK} doesn't need the boundary condition at the left boundary $\br \to 0$, as this is discussed in \citep{ItkinMuravey2020r} with a reference to Fichera theory (in other words, the PDE itself with substituted $\br = 0$ serves as the boundary condition). If, however, the boundary $\br = 0$ is attainable, the boundary condition at this point should be set as in \eqref{bc0}. This is also applicable to all below PDEs obtained from \eqref{PDEBK} by transformations.

It is worth noting that \eqref{PDEBK} is the stochastic Verhulst or stochastic logistic model, which are well-known in the population dynamics and epidemiology; see, eg., \citep{Verhulst, bacaer2011, Logistic2015} and references therein. In the past, several authors attempted to use this model in finance; see, eg., \citep{Chen2010, Londono2015, Halperin2018}. In our case, the stochastic Verhulst equation has the form
\begin{align} \label{verhulst}
d \br_t &= k(t) \br_t [\tth(t) - \br_t] dt + \sigma(t) \br_t dW_t, \\
\br_t &= [r_t - s(t)]/R, \qquad r(t=0) = r_0. \nonumber
\end{align}

\eqref{verhulst} can be explicitly solved (for the time-homogeneous coefficients this is done, eg., in \citep{Logistic2015}, Proposition~3.3). The following Proposition holds
\begin{proposition} \label{prop1}
The \eqref{verhulst} admits a unique positive solution $\br_t$
\begin{equation} \label{solV}
\br_t = \frac{1}{\kappa(t)} \frac{X_t}{\frac{T_0}{\kappa(0) \br_0} + \int_0^t X_q dq}, \qquad t \ge 0,
\end{equation}
\noindent where $X_t$ solves the lognormal SDE
\begin{equation} \label{eqS}
d X_t = \mu(t) X_t dt + \sigma(t) X_t d W_t, \qquad \mu(t) = \frac{\kappa'(t)}{\kappa(t)} + \kappa(t) \tth(t).
\end{equation}
Also,
\begin{enumerate}
\item The diffusion $\br_t$ is recurrent if and only if $q(t) \le 0$, where $q(t) = \frac{1}{2} - \frac{\kappa(t)\tth(t)}{\sigma^2(t)}, \ \forall t \in [0,T]$.

\item If $q(t) < 0, \ \forall t \in [0,T]$, assuming that the limits of $q(t)$ and $\sigma^2(t)/(2 \kappa(t)$ exist at $t \to \infty$,  the diffusion $\br_t$ converges in law towards the unique stationary Gamma probability distribution  $\Gamma\left(-2q(t),\sigma^2(t)/(2 \kappa(t)\right)_{t \to \infty}$.

\item If $\exists t_* \ge 0: q(t_*) > 0$, the diffusion goes a.s. to zero when time goes to infinity.
\end{enumerate}
\end{proposition}
\begin{proof}
The proof can be obtained by applying \Ito lemma to \eqref{solV} and using  \eqref{eqS}. The second part follows from Proposition 3.3 in \citep{Logistic2015}. It is interesting to note, that the condition $q < 0$ is precisely the Feller condition for the famous CIR model, \citep{andersen2010interest}.
\end{proof}

Thus, the stationary distribution for the Verhulst model is the Gamma distribution. It is easy to check that as compared with the mean-reversion lognormal model (the BK model) with the same parameters, the former has much fatter tails at the lower end, while the latter has the fatter tails when $\br_t \to \infty$. However, since, under the current market conditions, we are interested in modeling the lower end in the first place, the Verhulst model has a distinct advantage compared with the BK model. In other words, the probability of having lower rates for the Verhulst model is much higher than for the BK model, and comes naturally.

To illustrate this in a slightly different way, we produce a set of Monte-Carlo paths for both models which have the same volatility and mean-reversion rate, while the mean -reversion level $\bar{\theta}(t)$ in \eqref{BK1} is chosen as $\bar{\theta}(t) = 1 + \theta(t)$, so the dynamics \eqref{BK1} corresponds to the BK dynamics in \eqref{BK} for small $z_t$. The results obtained by using parameters given in Section~\ref{numE} are presented in Fig.~\ref{MC}, which shows that $r_{BK}$ is always higher than $r_{Vh}$, which confirms our theoretical observation in above.

\begin{figure}[!htb]
\caption{Typical paths of the difference $r_{BK} - r_{Vh}$ for the short-term BK and Verhulst interest rates as a function of time.}
\label{MC}
\begin{center}
\fbox{\includegraphics[totalheight=3.5in]{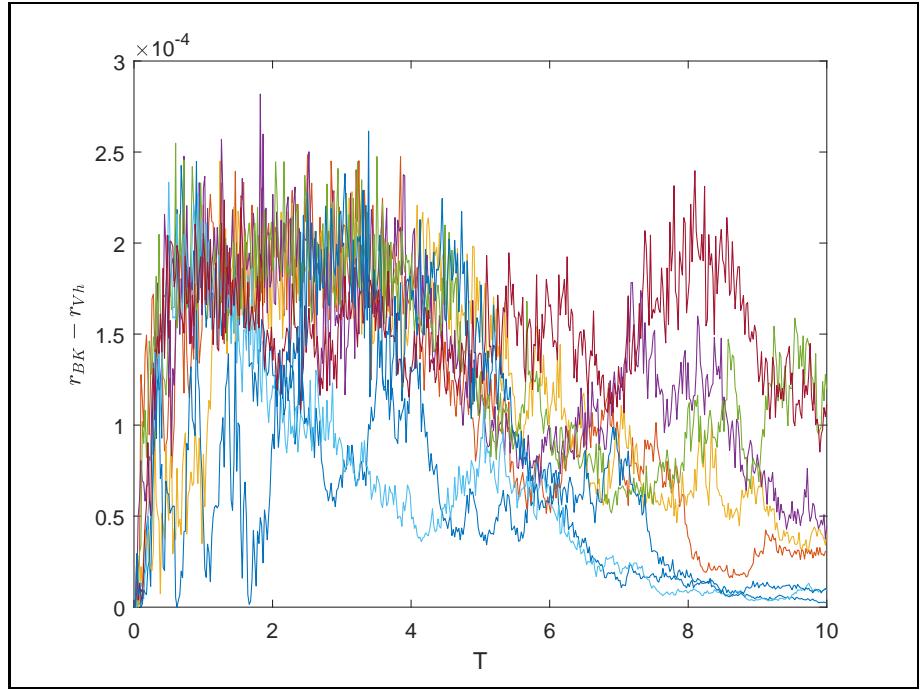}}
\end{center}
\end{figure}

Next, we aim to demonstrate that the Verhulst model is also more tractable than the BK model.

\section{An integral equation for the ZCB price in the Verhulst model} \label{intZCB}

In this section, we find the value of the ZCB by deriving and solving a Volterra integral equation of the second kind. We proceed with the elimination of the squared term in the drift in \eqref{PDEBK} via the following change of variables
\begin{equation}
   F(t, \br) = e^{\frac{\kappa(t)}{\sigma^2(t)} \br + \int_T^t s(k) dk} W(t,\br).
\end{equation}

This yields
\begin{align} \label{PDE_ZCB}
0 = \fp{W}{t} &+ \dfrac{1}{2}\sigma^2(t) \br^2 \sop{W}{\br} + \kappa(t) \tth(t) \br \fp{W}{\br}
+ \left[ \gamma(t) \br - \frac{\kappa^2(t)}{2\sigma^2(t)} \br^2  \right] W, \nonumber \\
\gamma(t) &= - R + \frac{\kappa^2(t) \tth(t)}{\sigma^2(t) } + \frac{d}{dt}\left(\frac{\kappa(t)}{\sigma^2(t)}\right), \nonumber \\
W(T,\br)  &= e^{-\frac{\kappa(T)}{\sigma^2(T)} \br}, \qquad
W(T,\br)\Big|_{\br \to \infty} = 0.
\end{align}
Another change of variables
\begin{align} \label{WPDE}
W(t, \br) &= u(\tau, y), \qquad y= \log \br + \int_0^t \left[ \frac{\sigma^2(k)}{2} - \kappa(k) \tth(k) \right] dk, \qquad \tau = \frac{1}{2}\int_t^T \sigma^2(k) dk
\end{align}
\noindent transforms this PDE into the following one
\begin{align} \label{PDE_ZCB_u}
\fp{u}{\tau} &=  \sop{u}{y} + \left[ \alpha(\tau) e^{y} + \beta(\tau) e^{2y}\right] u, \\
\alpha(\tau) &= \frac{2 \gamma(t(\tau))}{\sigma^2(t(\tau))} e^{-\int_0^t \left[ \frac{\sigma^2(k)}{2} - \kappa(k) \tth(k) \right] dk}, \qquad \beta(\tau) = - \frac{\kappa^2(t(\tau))}{\sigma^4(t(\tau))} e^{-\int_0^t \left[\sigma^2(k) - 2\kappa(k) \tth(k) \right] dk}. \nonumber
\end{align}
It is worth mentioning that, by the change of variables $ y \to -y$ and $\tau \to -\iu \tau$, this PDE  transforms into the time-dependent Schr\"{o}dinger equation with the unsteady Morse potential.

This PDE in \eqref{PDE_ZCB_u} should be solved subject to the initial and boundary conditions (see discussion in \citep{ItkinMuravey2020r})
\begin{align} \label{zcbTc}
u(0,y) &= \exp[-B(T) e^{y - A(T)}], \qquad u(\tau, y)\Big|_{y \to \infty} = 0, \\
B(T) &= \frac{\kappa(T)}{\sigma^2(T)}, \qquad A(T) = \int_0^T \left[ \frac{\sigma^2(k)}{2} - \kappa(k) \tth(k)\right] dk. \nonumber
\end{align}

Since $u(0,y), \  y \in (-\infty,\infty)$, $u(0,y)$ can be represented as
\begin{equation*}
u(0,y) = \frac{1}{2} - \frac{1}{2}(H(y) - H(-y)) + g(y), \qquad g(\infty) = g(-\infty) = 0,
\end{equation*}
\noindent where $H(y)$ is the Heaviside step function. The Fourier transforms of all parts in the RHS of this equation exist  implying that it also exists for $u(0,y)$. Therefore, applying the Fourier transform
\begin{equation} \label{FT_def}
\baru(\tau, \omega) = \int_{-\infty}^{\infty} u(\tau, y) e^{-i\omega y} dy,
\end{equation}
\noindent to both parts of \eqref{PDE_ZCB_u} yields the ordinary differential equation
\begin{align} \label{baru_ODE}
\frac{d \baru}{ d \tau} &+ \omega^2 \baru = g(\tau, \omega), \\
g(\tau, \omega) &= \int_{-\infty}^{\infty} u(\tau, y) \left[\alpha(\tau) e^{y} + \beta(\tau) e^{2y}\right] e^{-i\omega y} dy,
\qquad \baru(0, \omega) = \int_{-\infty}^{\infty} u(0, y) e^{-i\omega y} dy. \nonumber
\end{align}

The solution of this problem can be written as
\begin{equation}
\baru(\tau, \omega) =e^{-\omega^2 \tau} \baru(0,\omega) + \int_0^{\tau} e^{-\omega^2(\tau - s)} g(s, \omega) ds
\end{equation}

Now, applying the inversion formula
\[ u(\tau, y) = \frac{1}{2 \pi} \int_{-\infty}^{\infty} \baru(\tau, \omega) e^{i\omega y} d\omega,
\]
\noindent we obtain the following representation for $u(\tau, y)$
\begin{equation} \label{uRepr}
    u(\tau,y) = \frac{1}{2\pi} \int_{-\infty}^{\infty} \left[e^{-\omega^2 \tau} \baru(0,\omega) +
\int_0^{\tau} e^{-\omega^2(\tau - s)} g(s, \omega) ds \right] e^{i\omega y} d\omega
\end{equation}

Substituting the explicit representations for $\baru(0,\omega)$ and $g(\tau, \omega)$ into \eqref{uRepr}, and taking into account that the function $e^{-\omega^2\tau}$ is an even function, we obtain
\begin{align}
    u(\tau,y) = \frac{1}{\pi} \int_{0}^{\infty} \bigg\{ & e^{-\omega^2 \tau}  \int_{-\infty}^{\infty} u(0,\xi) \cos[\omega(y-\xi)] d\xi
\\ \nonumber
&+ \int_0^{\tau} ds \, e^{-\omega^2(\tau - s)} \int_{-\infty}^{\infty} u(\tau, \xi) \left[\alpha(s) e^{\xi} + \beta(s) e^{2\xi}\right]  \cos[\omega(y-\xi)] d\xi \bigg\}  d\omega.
\end{align}

Applying the identity (\citep{GR2007})
\[
\int_{-\infty}^\infty e^{-\beta x^2} \cos(b x) dx = \sqrt{\frac{\pi}{\beta}} \exp \left(-\frac{b^2}{4\beta} \right),
\]
\noindent and changing the order of integration, we get the integral equation for $u(\tau,y)$
\begin{equation} \label{IntegralEqZCB}
u(\tau,y) = \frac{1}{2\sqrt{\pi}} \Bigg\{ \int_{-\infty}^{\infty} \frac{u(0, \xi)}{\sqrt{\tau}} e^{-\frac{(y - \xi)^2}{4\tau}} d\xi +
 \int_{-\infty}^{\infty} d\xi \int_0^{\tau} \frac{u(k,\xi)}{\sqrt{\tau - k}} e^{-\frac{(y - \xi)^2}{4(\tau - k)}}
\left[\alpha(k) e^{\xi} + \beta(k) e^{2\xi}\right]  dk \Bigg\}.
\end{equation}

This is a two-dimensional Volterra equation of the second kind, see \citep{Lipton2001,LiptonKaush2020,ItkinMuravey2020r,CarrItkinMuravey2020} for the discussion. As mentioned by an anonymous referee, this equation is also a direct consequence of the Duhamel's principle.

As far as the numerical solution of \eqref{IntegralEqZCB} is concerned, the simplest scheme would be Picard iterations. First, we can change the order of integration, so the integral in time becomes the outer one. It can be approximated by using, eg. the trapezoidal rule. Then the coefficients of integration can be computed by using Fast Gauss Transform with the complexity $O(N M)$, $N$ is the number of grid points in $y$, $M$ - in $\tau$. And computing the outer integral requires $O(M^2)$ operations. Thus, if $M \le N$ the speed of the method is same as of a FD scheme with the second order of approximation in space and time (eg,, the Crank-Nicolson one). However, using higher order Simpson quadratures can reduce $M$ to $\sqrt{M}$ providing same accuracy, while doing same for the FD method is not trivial. We discuss this in more detail below in the paper.

\section{A closed-form solution for the ZCB price in the MBK model}

Here, we show that for some dependencies between the parameters of the model, the Cauchy problem  \eqref{PDE_ZCB} can be solved explicitly in terms of the Gauss hypergeometric function, \citep{as64}.

We start with the following change of variables
\begin{equation}
x = \phi(t) \br,
\qquad \phi(t) = e^{\int_0^t\left[C_\alpha \sigma^2(k) - \kappa(k) \tth(k)\right] dk},
\end{equation}
\noindent where $C_\alpha$ is a constant. This change of variables yields


This change of variables yields
\begin{align} \label{PDE_ZCB_2}
0 &= \fp{W}{t} + \dfrac{1}{2}\sigma^2(t) x^2 \sop{W}{x} + C_\alpha \sigma^2(t) x \fp{W}{x}
+ \left[ \frac{\gamma(t)}{\phi(t)} x - \frac{\kappa^2(t)}{2\sigma^2(t) \phi^2(t)} x^2 \right]W  \\
W(T,x)  &= e^{-\frac{\kappa(T)}{\sigma^2(T) \phi(T)} x }, \qquad
W(T,x)\Big|_{x \to \infty} = 0. \nonumber
\end{align}

We assume $\kappa(t)$, $\tth(t)$ and $\sigma(t)$ satisfy the following conditions
\begin{equation} \label{cond_for_WH_eq}
C_\gamma =  \frac{\gamma(t)}{\kappa(t)}, \qquad C_\sigma = \frac{\sigma^2(t) \phi(t)}{2 \kappa(t)},
\end{equation}
\noindent where $C_\gamma$ and $C_\sigma > 0$ are some constants. Using the definitions of $\gamma(t), \phi(t)$, and some algebra, we can show that under these conditions, we have
\begin{align}  \label{systemC1}
C_\gamma &= C_\alpha - \frac{R}{\kappa(t)}, \qquad C_\sigma =  \frac{\sigma^2(0)}{2\kappa(t)}, \qquad
\tth(t) = \frac{C_\alpha \sigma^2(t)}{\kappa(t)}+ \frac{2}{\kappa(t)} \frac{\sigma'(t)}{\sigma(t)}, \qquad \phi(t) = \frac{\sigma^2(0)}{\sigma^2(t)}.
\end{align}
%
%

The first equality implies that in this case the mean reversion rate $\kappa(t) = \kappa = const$, and $C_\gamma$ depends on $C_\alpha$ and $\kappa$, while $C_\sigma$ depends on $\kappa$ and $\sigma(0)$. It means that we can calibrate the model as follows. First, we calibrate the volatility term structure to the market, together with the constant mean reversion rate of $\kappa$ and the constant $ C_\alpha $. Second, we determine the time-dependent mean-reversion level $\tth(t)$  by using \eqref{systemC1}. Thus, in this version of the model, we have three calibration parameters: two of them - $\kappa$ and $C_\alpha$ are constants, and the normal volatility $\sigma(t)$ is time-dependent. In other words, this enables capturing the volatility term-structure of the market which seems to be the most important property, while assuming a constant mean reversion speed is not too restrictive. The time-dependence of the mean reversion level, however, is fully defined by $\sigma(t)$  and is corrected by another calibrated constant $C_\alpha$. So this seems to be a weak side of the model.

Now, applying another change of variables to \eqref{PDE_ZCB_2}
\begin{align} \label{anotherChange}
w &= \frac{x}{C_\sigma}, \qquad \tau = \frac{1}{2}\int_t^T \sigma^2(k) \, dk,
\qquad W(t,x) = w^{-C_\alpha}  e^{C_\alpha(1-C_\alpha)\tau} u(\tau,w),
\end{align}
\noindent we obtain the PDE with the time-homogeneous coefficients
\begin{equation} \label{PDEu_Whittaker}
0 =  \sop{u}{w} + \left[ - \frac{1}{4} + \frac{C_\gamma}{w}\right] u - \frac{1}{w^2}\fp{u}{\tau}.
\end{equation}
This PDE should be solved subject to the initial and boundary conditions
\begin{equation} \label{BC_whit}
u(0,w)  = w^{C_\alpha} e^{-\frac{\kappa(T)}{\sigma^2(T) \phi(T)} C_\sigma w} = w^{C_\alpha}e^{-w/2}, \qquad
u(\tau,w)\Big|_{w \to \infty} = 0.
\end{equation}
Applying the Laplace transform
\begin{equation} \label{LT_def}
\baru(\lambda,w) = \int_0^\infty e^{-\lambda \tau} u(w,\tau) d\tau.
\end{equation}
\noindent to \eqref{PDEu_Whittaker} and introducing $\mu = \sqrt{\lambda + 1/4}$, we obtain the following inhomogeneous ordinary differential equation
\begin{align} \label{whittaker_eq}
&\frac{d^2 \baru}{dw^2} + \left[ - \frac{1}{4} + \frac{C_\gamma}{w} + \frac{1/4 - \mu^2}{w^2}\right] \baru = -\frac{u(0,w)}{w^2}, \qquad \baru (\lambda,w)\Big|_{w \to \infty} = 0.
\end{align}

The corresponding homogeneous \eqref{whittaker_eq} is a Whittaker equation, which has two linearly independent solutions (the Whittaker functions) $M_{C_\gamma, \mu}(w)$ and $W_{C_\gamma, \mu}(w)$, \citep{as64}. A general solution of the problem \eqref{whittaker_eq} reads
\begin{align} \label{whitSol}
\baru(\lambda,w) &= C_1 M_{C_\gamma, \mu} (w) + C_2 W_{C_\gamma, \mu} (w) \\
&+ \frac{\Gamma\left(1/2 - C_\gamma + \mu\right)}{\Gamma\left(1 + 2\mu\right)}
\left(W_{C_\gamma, \mu} (w) \int_0^w \frac{u(0,\zeta)}{\zeta^2} M_{C_\gamma, \mu} (\zeta) d\zeta -
M_{C_\gamma, \mu} (w) \int_0^w \frac{u(0,\zeta)}{\zeta^2} W_{C_\gamma, \mu} (\zeta) d\zeta \right) \nonumber \\
&= C_3 M_{C_\gamma, \mu} (w) + C_2 W_{C_\gamma, \mu} (w) \nonumber \\
&+\frac{\Gamma\left(1/2 - C_\gamma + \mu\right)}{\Gamma\left(1 + 2\mu\right)}
\left(W_{C_\gamma, \mu} (w) \int_0^w \frac{u(0,\zeta)}{\zeta^2} M_{C_\gamma, \mu} (\zeta) d\zeta +
M_{C_\gamma, \mu} (w) \int_w^\infty \frac{u(0,\zeta)}{\zeta^2} W_{C_\gamma, \mu} (\zeta) d\zeta \right) \nonumber \\
C_3 &= C_1 + \int_0^\infty \frac{u(0,\zeta)}{\zeta^2} W_{C_\gamma, \mu} (\zeta) d\zeta. \nonumber
\end{align}

Using the asymptotic expressions for the Whittaker functions, \citep{as64}
\begin{align} \label{WhittakerAsymptotic}
w &\rightarrow 0: \ M_{\kappa, \mu}(w) = w^{\mu + 1/2} \left(1 + O(w) \right),\qquad
W_{\kappa, \mu}(w) = \frac{\Gamma(2\mu)}{\Gamma(1/2 +\mu- \kappa)} w^{1/2 -\mu} + O\left(w^{3/2 - \Re(\mu)}\right),  \\
w &\rightarrow \infty: \ M_{\kappa, \mu}(w) \sim \frac{\Gamma(1 + 2\mu)}{\Gamma(1/2 +\mu- \kappa)} e^{w/2} w^{-\kappa}, \qquad W_{\kappa, \mu}(w) \sim e^{-w/2} w^{\kappa}, \nonumber
\end{align}
\noindent and the boundary condition in \eqref{whittaker_eq}, we can set $C_3 = C_2 = 0$. Here $\Re(x)$ denotes the real part of $x$.

Since the integrands in \eqref{whitSol} have singularities at the points $w = 0$ and $w \to \infty$ we need to check that both functions $u(\tau, w), W(t,x)$ are regular at these points. Applying \eqref{WhittakerAsymptotic} and L'H\^{o}spital's rule yields
\begin{align*}
\lim_{w \to 0} \frac{\baru(\lambda, w)}{w^{C_\alpha}} &=
\lim_{w \to 0}  \Bigg[\frac{1}{2\mu}\frac{\int_0^w u(0,\zeta) \zeta^{-2}M_{C_\gamma, \mu} (\zeta) d\zeta}{w^{C_\alpha + \mu - 1/2}}
+ \frac{\Gamma\left(1/2 - C_\gamma + \mu\right)}{\Gamma\left(1 + 2\mu\right)}\frac{\int_w^\infty u(0,\zeta) \zeta^{-2}W_{C_\gamma, \mu} (\zeta) d\zeta}{w^{C_\alpha-\mu - 1/2}}
 \Bigg] \\
&= \frac{1}{2\mu} \lim_{w \to 0} \left[\frac{u(0,w) w^{-C_\alpha}}{C_\alpha + \mu - 1/2}  - \frac{u(0,w)w^{-C_\alpha}}{C_\alpha - \mu - 1/2} \right]
= \frac{ 1}{\left(\mu+C_\alpha - 1/2\right)\left(\mu - C_\alpha  + 1/2\right)} \lim_{w \to 0}  e^{-w/2}
\\
&= \frac{1}{\mu^2 - 1/4 - C_\alpha^2 + C_\alpha} = \frac{1}{\lambda - C_\alpha^2 + C_\alpha}
\\
\lim_{w \to \infty}  \frac{\baru(\lambda, w)}{w^{C_\alpha}}&=\lim_{w \to \infty} \left[\frac{\Gamma\left(1/2 - C_\gamma + \mu\right)}
{\Gamma\left(1 + 2\mu\right)}
\frac{\int_0^w u(0,\zeta) \zeta^{-2} M_{C_\gamma, \mu} (\zeta) d\zeta}{e^{w/2} w^{C_\alpha-C_\gamma}} + \frac{\int_w^\infty u(0,\zeta) \zeta^{-2} W_{C_\gamma, \mu} (\zeta) d\zeta}{e^{-w/2} w^{C_\alpha+C_\gamma}}
\right] \\
&=\lim_{w \to \infty} \left[ \frac{e^{w/2} w^{-C_\gamma-2} u(0, w) }{e^{w/2} w^{C_\alpha-C_\gamma}\left(1/2 +(C_\alpha- C_\gamma) / w \right)} - \frac{e^{-w/2} w^{C_\gamma -2} u(0, w)}{ e^{-w/2} w^{C_\alpha + C_\gamma}\left(-1/2 + (C_\alpha+C_\gamma) / w \right)} \right] \\
&= \lim_{w \to \infty} \left[ \frac{u(0,\omega)w^{-C_\alpha}}{ w/2 + w^2 (C_\alpha - C\gamma)}
- \frac{u(0,\omega)w^{-C_\alpha}}{ -w/2 + w^2 (C_\alpha + C\gamma)} \right] \\
&= \lim_{w \to \infty} \left[ \frac{e^{-w/2}}{ w/2 + w^2 (C_\alpha - C\gamma)}
- \frac{e^{-w/2}}{ -w/2 + w^2 (C_\alpha + C\gamma)} \right] =0.
\end{align*}
And inverting the Laplace transform  we obtain
\begin{align*}
\lim_{w \to 0} \frac{u(\tau, w)}{w^{C_\alpha}} = e^{-C_\alpha(1- C_\alpha)\tau}, \quad \lim_{w \to \infty} \frac{u(\tau, w)}{w^{C_\alpha}} =0.
\end{align*}

Returning back to the original variables yields
\begin{align*}
\lim_{x \to 0} W(t,x) = 1, \quad \lim_{x \to \infty} W(t,x) =0.
\end{align*}

Accordingly, since $\bar{r} \to 0$ implies $x \to 0$, this yields
\begin{equation} \label{bc0}
F(\tau,\bar{r})|_{\bar{r} \to 0} \to e^{\int_T^{t(\tau)} s(k) dk}.
\end{equation}
At $t=T$ (or $\tau=0$) this limit is consistent with the terminal condition in \eqref{termZCB}.

Expression \eqref{whitSol} can be further simplified by using the formula, \citep{GR2007}
\begin{align} \label{WM2I}
\int_0^{\infty} e^{-\frac{1}{2} (a_1+a_2) t \cosh(k)}  & \coth^{2\nu} \left( \frac{k}{2}  \right) I_{2\mu}
\left(t \sqrt{a_1 a_2} \sinh(k)\right) dk = \frac{\Gamma \left( 1/2 +\mu - \nu \right)}{ t \sqrt{a_1 a_2} \Gamma(1+2\mu)} W_{\nu,\mu} (a_1 t) M_{\nu,\mu} (a_2 t), \nonumber \\
&\Re\left( 1/2 + \mu - \nu \right) > 0, \qquad \Re(\mu) > 0, \qquad a_1 > a_2,
\end{align}
\noindent where $I_\nu(w)$ is the modified Bessel function. Therefore, setting $t=1, a_1 = w, a_2 = \xi$, and perceiving that \eqref{WM2I} is symmetric with respect to $a_1$ and $a_2$ (so that the integrals on $\xi$ in \eqref{whitSol} are complimentary and sum up to a single integral from 0 to infinity), we obtain
\begin{equation} \label{ubarSol}
\baru(\lambda,w) = \int_{0}^{\infty} \int_{0}^{\infty} u(0,\zeta)
\sqrt{\frac{w}{\zeta^3}}  e^{-\frac{1}{2} (w+ \zeta) \cosh(k)} \coth^{2C_\gamma} \left( \frac{k}{2}  \right) I_{2\mu} \left(\sqrt{w \zeta } \sinh(k) \right) dk \, d\zeta.
\end{equation}
Using the inversion formula for the Laplace transform, we get $u(\tau,w)$ of the form
\begin{equation} \label{inv_transform}
u(\tau,w) = \frac{1}{2\pi \iu}
\int_{\gamma_*} e^{\lambda \tau} \int_{0}^{\infty} \int_{0}^{\infty} u(0,\zeta)
\sqrt{\frac{w}{\zeta^3}} e^{-\frac{1}{2} (w+ \zeta) \cosh(k)} \coth^{2C_\gamma} \left( \frac{k}{2}  \right) I_{2 \mu} \left(\sqrt{w \zeta } \sinh(k)\right) dk \, d\zeta \, d\lambda.
\end{equation}
Here $\gamma_*$ denotes any vertical line $\Re(\lambda) = \gamma_*$ in the complex plane such that all singularities of the integrand in \eqref{inv_transform} lie to the left of this line.

Applying another identity, \citep{GR2007}
\begin{equation*}
\int_0^\infty x^{\nu - \frac 12 } e^{-\alpha x}
I_{2\mu} \left( 2 \beta \sqrt{x}\right) dx =  \frac{\Gamma(\nu + \mu + 1/2)}{\Gamma(2\mu + 1)} \beta^{-1} e^{\frac{\beta^2}{2\alpha}} \alpha^{-\nu} M_{-\nu, \mu} \left( \frac{\beta^2}{\alpha}\right), \qquad
\Re\left( \mu + \nu + \frac 12 \right) > 0,
\end{equation*}
\noindent to the internal integral of $\zeta$, we obtain the following representation for $u(\tau, w)$
\begin{align} \label{ZCB_LT}
u(\tau,w) &= \frac{e^{-w/2}}{2\pi \iu} \int_{\gamma_*} e^{\lambda \tau} \int_{0}^{\infty} \frac{2}{\sinh(k)}
\Bigg\{ \Upsilon(\lambda) e^{-\frac{w}{4} (\cosh(k) -1)} \coth^{2C_\gamma} \left( \frac{k}{2}  \right) \left( \frac{1 + \cosh(k)}{2} \right)^{1-C_\alpha} \\
&\times M_{1-C_\alpha, \sqrt{\lambda + 1/4}}\left(\frac{w}{2} (\cosh(k) - 1) \right)  \Bigg\} dk \, d\lambda,  \qquad
\Upsilon(\lambda) = \frac{\Gamma\left(\sqrt{\lambda + 1/4} +C_\alpha- 1/2 \right)}{\Gamma\left(2\sqrt{\lambda + 1/4}  + 1\right)}. \nonumber
\end{align}

By changing the variable of integration $k$ in \eqref{ZCB_LT} as
\begin{equation*}
2 \log\left[ \tanh\left( \frac k2 \right)\right] = -\nu, \qquad \frac{2 dk}{\sinh(k)} = -d\nu, \qquad \frac{\cosh(k) - 1}{2} = \frac{1}{ e^\nu - 1},
\end{equation*}
\noindent we get
\begin{align} \label{ZCB_LT2}
u(\tau,w) = \frac{e^{-w/2}}{2\pi \iu}  \int_{\gamma_*} e^{\lambda \tau} \Upsilon(\lambda)
\int_{0}^{\infty} e^{-\frac{w}{2(e^\nu - 1)} + \nu C_\gamma}  M_{1-C_\alpha, \sqrt{\lambda + 1/4}}\left(\frac{w}{e^\nu - 1}  \right) \left(\frac{e^\nu }{e^\nu - 1} \right)^{1-C_\alpha} d\nu d\lambda.
\end{align}

This integral has poles at the points $\lambda = \lambda_k$
\begin{equation}
\label{residualsPoints}
\sqrt{\lambda_k  + \frac{1}{4}} = \mu_k = \frac{1}{2} - C_\alpha - k, \quad k = 0,\ldots,\left\lfloor\frac 12 - C_\alpha \right\rfloor,
\end{equation}
\noindent since the Gamma function in the numerator of $\Upsilon(\lambda)$ turns to complex infinity when its argument is a non-positive integer or zero. The number of poles depends on the value of $C_\alpha$ : if $C_\alpha \geq 1/2$, the integrand function has no poles, if $C_\alpha <1/2$ the number of the poles is $K = 1 + \left\lfloor\frac 12 - C_\alpha \right\rfloor$, where $\lfloor x \rfloor$ is the floor of $x$. The function $\sqrt{\lambda + 1/4}$ is a multivalued function of $\lambda$, i.e., the point $\lambda = -1/4$ is a branching point. Therefore, let us construct the contour of integration in \eqref{inv_transform} as the so-called keyhole contour presented in Fig.~\ref{contour}.
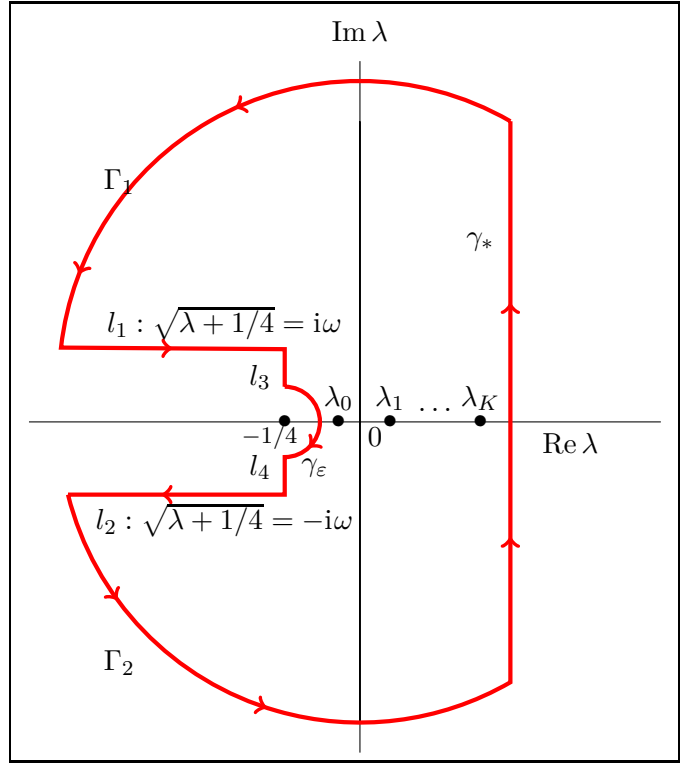
\begin{figure}[!htb]
\hspace{-0.5in}
\captionsetup{width=0.8\linewidth}	
\begin{center}
\caption{Contour of integration of \eqref{inv_transform} in the complex $\lambda$ plane with poles at
$\lambda_0,\lambda_1,\ldots,\lambda_K$. This picture corresponds to $C_\alpha > 0$, and $C_\alpha < 1/2$. If $C_\alpha < 0$, all poles are positive.}
\label{contour}
		\fbox{
			\begin{tikzpicture}
			\begin{scope}[scale=4.]
			\def\bigradius{1}
			\def\axisradius{1}
			\def\gammaradius{0}
			\def\omegaradius{0}
			\def\littleradius{0.117}
			\def\shift{1.4}

			\def\xii{14}
            \def\zii{30}
			\def\ang{6.5}
			\def\sh{0.07}
			\def\pole{-0.25}
			\def\alpha{35}
			\draw (-1.1*\axisradius, 0) -- (1.*\axisradius,0)
			(\omegaradius, -1.1*\axisradius) -- (\omegaradius, 1.2*\axisradius)
			(\gammaradius, -\bigradius) -- (\gammaradius, \bigradius);
					
			\draw[red, ultra thick, decoration={ markings,
  	 		  	 mark=at position 0.05 with {\arrow{>}}
  	 		  	,mark=at position 0.13 with {\arrow{>}}	  	
				,mark=at position 0.3 with {\arrow{>}}
				,mark=at position 0.4 with {\arrow{>}}
				,mark=at position 0.6 with {\arrow{>}}
				,mark=at position 0.7 with {\arrow{>}}
			    ,mark=at position 0.78 with {\arrow{>}}
				,mark=at position 0.88 with {\arrow{>}}
				,mark=at position 0.96 with {\arrow{>}}
			},
			postaction={decorate}]
			let
			\n1 = {asin(\sh/\littleradius)},
			\n2 = \littleradius,
			\n3 = \pole,
			\n4 = {\bigradius*cos(\xii)},
			\n5 = {\bigradius*sin(\xii)},
			\n6 = {\bigradius*sin(-90)},
			\n7 = {\bigradius*cos(\zii)},
			\n8 = {\bigradius*sin(\zii)}
			in
			(-\n4,-\n5) arc(-180+\xii:-90+\zii:\bigradius) 			
			-- (\n8, -\n6)
			(\n8,-\n6) arc(90 - \zii:180-\ang:\bigradius)
			--(\n3, \n5)
			--(\n3, \n2)
			(\n3, \n2) arc(90: -90:\littleradius)
	  		--(\n3, -\n5)
	  		--(-\n4,-\n5)
			;			
			\node at (0.7*\axisradius,-0.07){$\operatorname{Re} \lambda$};
			\node at (0,1.3*\axisradius) {$\operatorname{Im} \lambda$};
			\node at (0.05,-0.05){$0$};
			\node at (\pole-0.05,-0.05){{\footnotesize $-1/4$}};
			\node at (0.4,0.6) {$\gamma_*$};
			\node at (-\pole-0.4,-1.3*\littleradius) {$\gamma_{\varepsilon}$};
			\node at (-1/4,0) {$\bullet$};
			\node at (-1/14,0) {$\bullet$};
			\node at (0.1,0) {$\bullet$};
			\node at (0.4,0) {$\bullet$};
			\node at (-1/14,0.08) {$\lambda_0$};
			\node at (0.1,0.08) {$\lambda_1$};
            \node at (0.25,0.04) {$\ldots$};
			\node at (0.4,0.08) {$\lambda_K$};

			\node at (-0.8,0.8) {$\Gamma_1$};
			\node at (-0.8,-0.8) {$\Gamma_2$};
			\node at (-0.45, 2.8*\littleradius) {$l_1 : \sqrt{\lambda + 1/4}= \iu \omega$};
			\node at (-0.45,-2.8*\littleradius){$l_2 :\sqrt{\lambda + 1/4}= -\iu \omega$};
			\node at (-0.7*\littleradius + \pole,  0.15) {$l_3$};
			\node at (-0.7*\littleradius + \pole, -0.15) {$l_4$};
			\end{scope}
			\end{tikzpicture}
		}
	\end{center}
\end{figure}
In more detail, this contour can be described as follows. It starts with a vertical line $\gamma_*$ extending to the two big symmetric arcs $\Gamma_1$ and $\Gamma_2$ around the origin with the radius $R_1$; connecting to two horizontal parallel lines segments $l_1, l_2$ at $\operatorname{Im}(\sqrt{\lambda + 1/4}) = \pm \omega$; then extending to two vertical line segments $l_3, l_4$  which end points are connected to the semi-arc $\gamma_\varepsilon$ with the radius $\varepsilon$ around the point $\operatorname{Re}(\lambda) = -1/4$. Using a standard technique, we take the limit $\varepsilon \to 0, R_1 \to \infty$, so in this limit the integrals along the lines $l_3$ and $l_4$ are cancelled out. The integral along the contours $\Gamma_1$ and $\Gamma_2$ tends to zero if $R_1 \to \infty $ due to Jordan's lemma. Hence, according to the Cauchy residue theorem, the sum of integral along the vertical line $\gamma_*$ and two integrals along the horizontal semi-infinite lines $l_1$ and $l_2$ is equal to the sum of residuals.

Let us define the sum of residuals as $\calR(\tau, z)$ and the sum of integrals along the lines $l_1$ and $l_2$ with a negative sign as $\calI(\tau, z)$. We explicitly compute them in the next section.

\subsection{Calculation of residuals}

Using the well-known expressions for the poles of the Gamma function, and the connection formula for the Whittaker functions $M_{k, \mu}(z)$ and $W_{k,\mu}(z)$, \citep{as64}
\begin{equation*}
\underset{x=-n}\Res \Gamma(x) = \frac{(-1)^n}{n!}, \quad
W_{\frac 12 a + \frac 12 + n, \frac 12 a}(z) = (-1)^n \frac{\Gamma(a + 1 + n)}{\Gamma(a + 1)} M_{\frac 12 a + \frac 12 + n, \frac 12 a} (z),
\end{equation*}
\noindent we obtain
\begin{align} \label{residual_value}
\underset{\lambda = \lambda_k} \Res  \left\{
e^{\lambda \tau}\Upsilon(\lambda) M_{1-C_\alpha, \sqrt{\lambda + 1/4}} \left( \frac{w}{e^\nu - 1}\right) \right\}
&=\frac{2 \mu_k }{k!} \frac{e^{\left(\mu_k^2 - \frac{1}{4}\right) \tau} }{\Gamma\left( \frac 32 -C_\alpha + \mu_k\right)}  W_{1-C_\alpha, \mu_k} \left( \frac{w}{e^\nu - 1}\right).
\end{align}
Thus, the sum of the residuals after the substitutions $1 /(e^{\nu} -1) \mapsto \varsigma, \,  \mu_k = \frac 12 - k - C_\alpha$ reads
\begin{align} \label{res1}
\calR(\tau,w) &= e^{-w/2}\sum_{k = 0}^K  \frac{2 \mu_k}{k!} \frac{e^{\left(\mu_k^2 - \frac{1}{4}\right)\tau}}{\Gamma\left( \frac 32 -C_\alpha + \mu_k\right)}  \int_0^\infty W_{1-C_\alpha, \mu_k} \left( w \varsigma \right) e^{-\frac{w \varsigma}{2}} \left( 1+ \varsigma \right) ^{C_\gamma - C_\alpha} \varsigma^{-1- C_\gamma} d\varsigma.
\end{align}
The integrals of $\varsigma$ can be computed analytically
\begin{align} \label{MathIntviaGammaDef}
\int_0^\infty W_{1-C_\alpha, \mu} \left( w \varsigma \right) e^{-\frac{w \varsigma}{2}} \left( 1+ \varsigma \right) ^{C_\gamma - C_\alpha} \varsigma^{-1- C_\gamma} d\varsigma
&= w^{\frac 12 + \mu} \frac{\Gamma_{C_\gamma}^+\left(\iu \mu\right) \Gamma_{C_\gamma}^-\left(\iu \mu\right)}{\Gamma(C_\alpha - C_\gamma)} U\left(\frac{1}{2} + \mu - C_\gamma, 2\mu + 1, w\right), \nonumber \\
\Gamma^{\pm}_{y} (\omega) = \Gamma\left(\frac 12 -y \pm \iu \omega\right), \qquad
\Re(C_\gamma) &< \frac{1}{2} + \Re(\mu), \quad \Re(C_\gamma) < \frac{1}{2} - \Re(\mu), \quad w > 0.
\end{align}

Here $U(a,b,x)$ is a Kummer confluent hypergeometric function, \citep{as64}. Using the relation between the Kummer and Whittaker functions
\begin{equation*}
W_{k, \mu}(x) = e^{-\frac{z}{2}} z^{1/2 + \mu} U\left(\frac 12 + \mu - k, 1 + 2\mu, x\right),
\qquad W_{k, \mu}(x) = W_{k,-\mu}(x),
\end{equation*}
\noindent  we finally obtain from \eqref{res1}, \eqref{MathIntviaGammaDef}
\begin{align} \label{Psi_value}
\calI(\tau,w)=\sum_{k = 0}^K \frac{e^{\left(\mu_k^2 - \frac{1}{4}\right) \tau} }{k!}
\frac{2\mu_k }{\Gamma\left( \frac 32 -C_\alpha + \mu_k\right)} \frac{\Gamma_{C_\gamma}^+\left(\iu \mu_k\right) \Gamma_{C_\gamma}^-\left(\iu \mu_k\right)}{ \Gamma(C_\alpha - C_\gamma)} W_{C_\gamma, \mu_k}\left(w\right).
\end{align}

\subsection{Calculation of the integrals at different branches}

The integrals along the lines $l_1$ and $l_2$ read
\begin{align*}
&\int_{l_1}
e^{\lambda \tau} \Upsilon(\lambda) M_{1-C_\alpha, \sqrt{\lambda + 1/4}}\left(\frac{w}{e^\nu - 1}  \right) d\lambda
= 2 e^{-\tau/4} \int_0^\infty \omega e^{-\omega^2\tau}
\frac{\Gamma\left(\iu\omega + C_\alpha- 1/2 \right)}{\Gamma\left(2\iu\omega  + 1\right)} M_{1- C_\alpha, \iu\omega}\left(\frac{w}{e^\nu - 1}  \right) d\omega, \\
&\int_{l_2} e^{\lambda \tau} \Upsilon(\lambda) M_{1-C_\alpha, \sqrt{\lambda + 1/4}}\left(\frac{w}{e^\nu - 1}  \right) d\lambda = -2 e^{-\tau/4} \int_0^\infty \omega e^{-\omega^2\tau}
\frac{\Gamma\left(-\iu\omega + C_\alpha - 1 /2\right)}{\Gamma\left(-2\iu\omega  + 1\right)} M_{1-C_\alpha, -\iu\omega} \left(\frac{w}{e^\nu - 1}  \right)d\omega.
\end{align*}

Applying the connection formula between the Whittaker functions $M_{\nu, \mu}(w)$ and $W_{\nu, \mu}(w)$,  and the Euler's reflection formulas, \citep{as64}
\begin{align*}
W_{\lambda, \mu}(w) &= \frac{\Gamma(-2\mu)}{\Gamma(1/2 - \mu - \lambda)}M_{\lambda, \mu}(w) +
\frac{\Gamma(2\mu)}{\Gamma(1/2 + \mu - \lambda)}M_{\lambda, -\mu}(w) \\
\Gamma(w) \Gamma(1-w) &= \frac{\pi}{\sin(\pi w)}, \qquad \Gamma(1/2 + w) \Gamma(1/2-w) = \frac{\pi}{\cos(\pi w)}
\end{align*}
\noindent we get the following expression for the sum of the integrals along the lines $l_1$ and $l_2$
\begin{align*}
\int_{l_1 + l_2} & e^{\lambda \tau} \Upsilon(\lambda) M_{1-C_\alpha, \sqrt{\lambda + 1/4}}\left(\frac{w}{e^\nu - 1}  \right) d\lambda=  2e^{-\tau/4} \int_0^\infty \frac{\omega e^{-\omega^2\tau}}{2 \iu \omega}
\Bigg\{ \frac{\Gamma\left(\iu\omega +C_\alpha- 1/2\right)}{\Gamma\left(2\iu\omega \right)} M_{1-C_\alpha, \iu\omega}\left(\frac{w}{e^\nu - 1}  \right) \\
&+ \frac{\Gamma\left(-\iu\omega  +C_\alpha - 1 /2\right)}{\Gamma\left(-2\iu\omega \right)} M_{1-C_\alpha, -\iu\omega} \left(\frac{w}{e^\nu - 1}  \right) \Bigg\} d\omega \\
&= 2e^{-\tau/4} \int_0^\infty \frac{\omega e^{-\omega^2\tau}}{2 \iu \omega}
\frac{\Gamma_{1-C_\alpha}^+\left(\omega\right) \Gamma_{1-C_\alpha}^-\left(\omega\right)} {\Gamma\left(2\iu\omega \right)\Gamma\left(-2\iu\omega \right)} W_{1-C_\alpha, \iu\omega}\left(\frac{w}{e^\nu - 1}  \right) d\omega \\
&=-\frac{2\iu}{\pi} e^{-\tau/4} \int_0^\infty \omega e^{-\omega^2\tau} \sinh\left(2 \pi \omega \right)
\Gamma_{1-C_\alpha}^+\left(\omega\right) \Gamma_{1-C_\alpha}^-\left(\omega\right)
W_{1-C_\alpha, \iu\omega}\left(\frac{w}{e^\nu - 1}  \right) d\omega.
\end{align*}
Substituting this expression into \eqref{ZCB_LT2}, we obtain
\begin{align} \label{Phi_value}
\calI(\tau, w) &= \frac{ e^{-w/2 - \tau/4}}{\pi^2} \int_0^\infty \int_0^\infty \omega e^{-\omega^2\tau}
\sinh\left(2 \pi \omega \right) \Gamma_{1-C_\alpha}^+\left(\omega\right) \Gamma_{1-C_\alpha}^-\left(\omega\right)
W_{1-C_\alpha, \iu\omega}\left(\frac{w}{e^\nu - 1}  \right) e^{-\frac{w}{2(e^\nu - 1)}} \\
&\times \frac{e^{\nu(1-C_\alpha+ C_\gamma) }}{\left(e^\nu - 1\right)^{1-C_\alpha}}  d\nu \, d \omega  = \frac{1}{\pi^2 \Gamma (C_\alpha-C_\gamma)} \int_0^\infty
\omega e^{-(\omega^2  +1/4) \tau} \sinh\left(2 \pi \omega\right) \Gamma_{C_\gamma}^-(\omega) \Gamma_{C_\gamma}^+(\omega) \nonumber \\
&\times \Gamma_{1-C_\alpha}^-(\omega) \Gamma_{1-C_\alpha}^+(\omega) W_{C_\gamma, \iu \omega}\left(w\right) \, d \omega. \nonumber
\end{align}

\subsection{A closed-form solution for the ZCB price}

Combining \eqref{Phi_value} and \eqref{Psi_value}, we obtain $u(\tau,w)$ in closed-form
\begin{align} \label{u_explicit_f4}
u(\tau,w) =& \sum_{k = 0}^K \frac{e^{\left(\mu_k^2 - \frac{1}{4}\right) \tau} }{k!}
\frac{2 \mu_k}{\Gamma\left( \frac 32 -C_\alpha + \mu_k\right)}
\frac{\Gamma_{C_\gamma}^-\left(\iu \mu_k\right) \Gamma_{C_\gamma}^+\left(\iu \mu_k\right)}{ \Gamma(C_\alpha - C_\gamma)} W_{C_\gamma, \mu_k}\left(w\right)
\\ \nonumber
&+ \frac{1}{\pi^2 \Gamma (C_\alpha-C_\gamma)} \int_0^\infty
\omega e^{-(\omega^2  +1/4) \tau} \sinh\left(2 \pi \omega\right) \Gamma_{C_\gamma}^-(\omega) \Gamma_{C_\gamma}^+(\omega) \Gamma_{1-C_\alpha}^-(\omega) \Gamma_{1-C_\alpha}^+(\omega)
W_{C_\gamma, \iu \omega}\left(w\right) \, d \omega.
\end{align}

Since $u(0,w) = w^{C_\alpha}e^{-w/2}$, we obtain a previously unknown identity
\begin{align} \label{intNew}
\frac{1}{\pi^2}&\int_0^\infty
\omega \sinh\left(2 \pi \omega\right) \Gamma_{C_\gamma}^-(\omega) \Gamma_{C_\gamma}^+(\omega) \Gamma_{1-C_\alpha}^-(\omega) \Gamma_{1-C_\alpha}^+(\omega)
W_{C_\gamma, \iu \omega}\left(w\right) \, d \omega \\
&= w^{C_\alpha}e^{-w/2} \Gamma (C_\alpha-C_\gamma) -
\sum_{k = 0}^K \frac{1 }{k!}
\frac{2 \mu_k}{\Gamma\left( \frac 32 -C_\alpha + \mu_k\right)}
\Gamma_{C_\gamma}^-\left(\iu \mu_k\right) \Gamma_{C_\gamma}^+\left(\iu \mu_k\right) W_{C_\gamma, \mu_k}\left(w\right) \nonumber
\end{align}
\noindent which could be verified by numerical integration. Thus, \eqref{systemC1}, \eqref{u_explicit_f4} can also be represented in the form
\begin{align} \label{u_explicit_f4_new}
&F(\tau,w) e^{C_a(C_a-1)\tau - \int_T^{t(\tau)} s(k) dk} = 1 +e^{w/2} w^{-C_\alpha} \Bigg(
\sum_{k = 0}^K
\frac{2 \mu_k \left[ e^{\left(\mu_k^2 - \frac{1}{4}\right) \tau} - 1 \right]}{ k! \Gamma\left( \frac 32 -C_\alpha + \mu_k\right)}
\frac{\Gamma_{C_\gamma}^-\left(\iu \mu_k\right) \Gamma_{C_\gamma}^+\left(\iu \mu_k\right)}
{\Gamma(R/\kappa)} W_{C_\gamma, \mu_k}\left(w\right) \nonumber \\
&+ \frac{1}{\pi^2 \Gamma (R/\kappa)} \int_0^\infty \omega \left[e^{-(\omega^2  +1/4) \tau} - 1 \right] \sinh\left(2 \pi \omega\right) \Gamma_{C_\gamma}^-(\omega) \Gamma_{C_\gamma}^+(\omega) \Gamma_{1-C_\alpha}^-(\omega) \Gamma_{1-C_\alpha}^+(\omega) W_{C_\gamma, \iu \omega}\left(w\right) \, d \omega \Bigg).
\end{align}

Also, it is known that $\Gamma(a + b \iu) \Gamma(a - b \iu) \in \mathbb{R}$, \citep{Cohen1940}. Since, \citep{Temme1978}
\begin{align*}
W_{k,\mu}\left(\frac{1}{2}w^2\right) &= 2^{-k}\sqrt{w} \sum_{n=0}^\infty \frac{1}{n!}
\frac{\Gamma(\frac{1}{2} + 2 \mu + n)}{\Gamma(\frac{1}{2} + 2 \mu-n)} \frac{D_{2k-n-1/2}(w)}{(2 w)^n}, \qquad D_\nu(x) = 2^{\nu/2} U\left(-\frac{1}{2}\nu, \frac{1}{2}, \frac{1}{2} x^2\right),
\end{align*}
\noindent we have
\begin{align*}
\frac{\Gamma\left(\frac{1}{2} + 2 \mu + n\right)}{\Gamma\left(\frac{1}{2} + 2 \mu-n\right)} &=
\Gamma\left(\frac{1}{2} + 2 \mu + n\right) \frac{\cos(\pi (2\mu - n))}{\pi} \Gamma\left(\frac{1}{2} - 2 \mu + n \right) \in \mathbb{R}, \quad \mu = \iu \omega,
\end{align*}
\noindent and so $W_{C_\gamma, \iu \omega}\left(w\right) \in \mathbb{R}$. Thus, in \eqref{u_explicit_f4}  $\Gamma_{C_\gamma}^-(\omega) \Gamma_{C_\gamma}^+(\omega) \Gamma_{1-C_\alpha}^-(\omega) \Gamma_{1-C_\alpha}^+(\omega) W_{C_\gamma, \iu \omega}\left(w\right) \in \mathbb{R}$.

\section{Numerical examples} \label{numE}

To validate our analytical solution, we compute the ZCB prices by using numerical integration of \eqref{u_explicit_f4_new}, where the explicit form of the parameter $\sigma(t)$ is
\begin{equation} \label{ex}
\sigma^2(t) = \sigma_a + \frac{\sigma_b}{t + \sigma_c},
\end{equation}
\noindent where $\sigma_a, \sigma_b, \sigma_c$ are constants. We also assume that $s(t) = 0$, and so $R = r_0$. With these assumptions, we have
\begin{align} \label{exCoef}
\tth = \frac{C_\alpha \left(\sigma_a (\sigma_c + t) + \sigma_b\right) }{\kappa (\sigma_c + t)}  -\frac{\sigma_b}{k (\sigma_c+t) [\sigma_a (\sigma_c+t)+\sigma_b]}, \qquad
\tau = \frac{1}{2} \left[ \sigma_a t + \sigma_b \log\left(\frac{t + \sigma_c}{\sigma_c}\right)\right]
\end{align}
Since we are interested in the positive values of $\tth(t)$, it implies $\sigma_b < 0$.

We present the model parameters for this test in Table~\ref{tab1} (the parameters are artificial and not obtained by calibration to the market). We run the test for a set of maturities $T \in [1/12, 0.3,0.5,1,2,5, 10,20,30,50]$ years. We show the ZCB prices computed in our numerical experiment in Fig~\ref{ZCBprice}. As a benchmark, we use the numerical solution of \eqref{PDEBK} obtained by the FD method, \cite{ItkinBook}, and the solution of \eqref{BK1} obtained by Monte Carlo. To accelerate the FD solution, instead of \eqref{PDEBK}, we solve the forward equation for the density and then find the prices by integrating the payoff with the corresponding density. The FD solver runs on a non-uniform grid with 100 nodes in space $w$ and 200 steps in time $t$. However, for long maturities, more nodes in space might be necessary, see \citep{ItkinBook} in more detail. The Monte Carlo method uses 500,000 paths and 500 steps in time (so for $T=50$ years there is some bias in the results due to a small number os time steps). We perform all calculations in Matlab.
\begin{table}[!htb]
\begin{center}
\caption{Parameters of the test.}
\label{tab1}
\begin{tabular}{|c|c|c|c|c|c|}
\hline
$r_0$ & $\kappa_0$ & $\sigma_a$ & $\sigma_b$ & $\sigma_c$  & $C_\alpha$ \\
\hline
0.03 & 2.0 & 0.64 & -1.0 & 5.0  & 0.3 \\
\hline
\end{tabular}
\end{center}
\end{table}

\begin{figure}[!htb]
\vspace{-0.1in}
\caption{ZCB prices computed in the test by using Analytic, FD and Monte Carlo solutions.}
\label{ZCBprice}
\begin{center}
\fbox{\includegraphics[totalheight=3.5in]{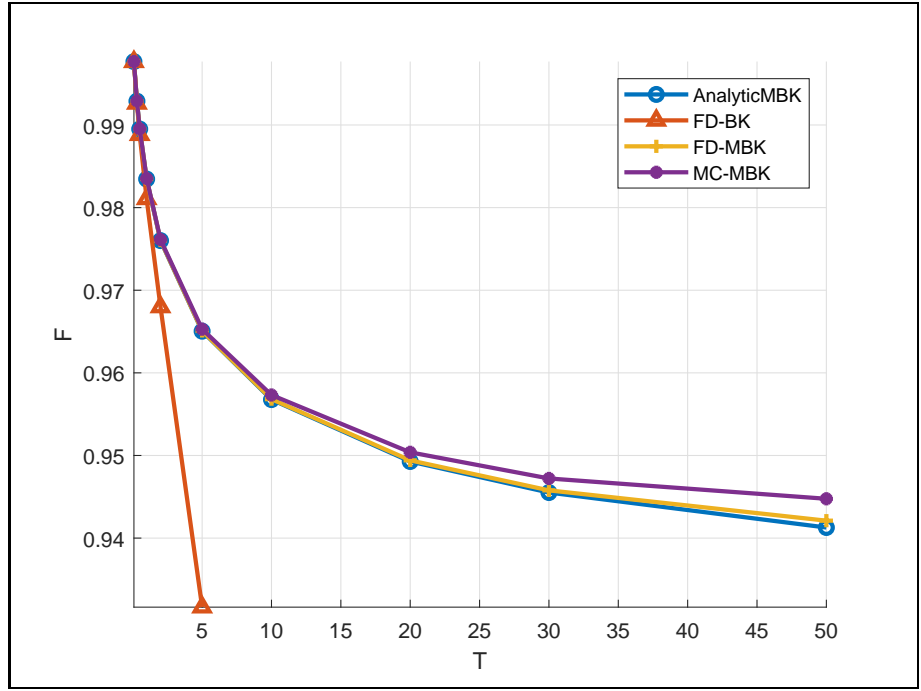}}
\end{center}
\end{figure}

\begin{table}[htbp]
  \centering
  \caption{Relative error in bps between Analytic, FD and MC solutions in the test.}
  \label{tab:relErr}%
\small
    \begin{tabular}{|l|r|r|r|r|r|r|r|r|r|r|}
    \toprule
          & \multicolumn{10}{c|}{T, years} \\
    \midrule
    Rel error & 0.0833 & 0.3   & 0.5   & 1     & 2     & 5     & 10    & 20    & 30    & 50 \\
    \hline
    Anal - FD & -0.0385 & -0.0625 & -0.0713 & -0.0929 & -0.2377 & -0.2818 & -0.5767 & -1.4387 & -2.8153 & -8.7802 \\
    \hline
    Anal - MC & -0.0844 & -0.2447 & -0.2803 & -0.6167 & -1.1802 & -2.5238 & -5.8173 & -11.7015 & -18.0272 & -36.9372 \\
    \bottomrule
    \end{tabular}%
\end{table}%

The results obtained by all methods coincide with high accuracy. We compare the corresponding relative errors of those results in Table~\ref{tab:relErr}. Also, we compute the ZCB prices in the BK model for small $T$, so that $e^{z_t} \approx 1 + z_t$. We choose $\bar{\theta}(t)_{BK} = \bar{\theta}(t)_{MBK} -1$. We show the result in Fig~\ref{ZCBprice} as well. It can be seen that the BK ZCB prices agree with the corresponding MBK ZCB prices for $T < 3$, which is due to the fact that $z$ is small.

As far as the performance of the methods is concerned, the elapsed time for computing all 10 ZCB bond prices by using the FD methods is 130 msec. For the analytics, since the only term under the integral in \eqref{u_explicit_f4_new}, which depends on $\tau$, is $e^{-(\omega^2 + 1/4)\tau}$, the other terms, including complex-valued Gamma and Whittaker functions, can be computed just once at the beginning of the script and then re-used. We calculate the integral by using the Simpson rule with 75 nodes; the elapsed time for getting all 10 ZCB prices is 55 msec. We note that internal Matlab implementation of  the Whittaker functions is a bit slow as it relies on Simulink to compute them. In other programming languages, eg., C++ or python this is not an issue. Anyway, the performance of our method is on par with that of the forward FD solver, while the accuracy is higher. Moreover, we can further improve the accuracy of the integration by using higher-order quadratures while keeping the elapsed time similar. At the same time, for the FD method, this is problematic (but perhaps can be done by using Radial Basis Functions methods).

\section{Conclusion}

To summarize our findings,  we have shown that in the current market environment, it is necessary to update the classical short-rate models. We introduced a useful extension of the popular BK model (the Verhulst model), which naturally produces prolonged periods of low rates and is more tractable. Finally, several complementary numerical and analytical methods to efficiently compute prices of the ZCB have been derived.

\vspace{0.2in}

\vspace{0.4in}
\appendixpage
\appendix
\numberwithin{equation}{section}
\setcounter{equation}{0}

\section{Stylized facts about the BK model}

The Black-Karasinski (BK) model was introduced in \citep{BK1991}, see also \citep{BM2006} for a more detailed discussion. The BK is a one-factor short-rate model of the form
\begin{align} \label{BK}
d z_t &= k(t)[\theta(t) - z_t] dt + \sigma(t)dW_t, \qquad r \in \mathbb{R}, \ t \ge 0, \\
r_t &= s(t) + R e^{z_t}, \qquad r(t=0) = r_0. \nonumber
\end{align}
Here $t$ is the time, $r_t$ is the short interest rate, $\kappa(t) > 0$ is the constant speed of mean-reversion, $\theta(t)$ is the mean-reversion level, $\sigma(t)$ is the volatility, $R$ is some constant  with the same dimensionality as $r_t$, eg., it can be 1/(1 year), $W_t$ is the standard Brownian motion. This model is similar to the Hull-White model but preserves the positivity of $r_t$ by exponentiating the OU random variable $z_t$. Frequently, practitioners add a deterministic function (shift) $s(t)$ to the definition of $r_t$ to address possible negative rates and be more flexible when calibrating the term-structure of the interest rates.

By the \Ito lemma the short rate $\bar{r}_t = (r_t - s(t))/R$ in the BK model solves the following stochastic differential equation (SDE)
\begin{equation} \label{BKr}
d \bar{r}_t = [k\theta(t) + \frac{1}{2}\sigma(t)^2 - k \log \bar{r}_t] \bar{r}_t dt + \sigma(t) \bar{r}_t dW_t.
\end{equation}
\noindent This SDE can be explicitly integrated. Let $0 \le s \le t \le T$, with $T$ being the maturity of a ZCB. Then $r_t$ can be represented as, \citep{BM2006}
\begin{equation} \label{BKsol}
\bar{r}_t = \exp\left[
e^{-k(t-s)} \log \bar{r}_s + k \int_s^t e^{-k(t-u)} \theta(u) du + \int_s^t \sigma(u) e^{-k(t-u)} dW_u \right],
\end{equation}
\noindent and thus, conditionally on filtration $\mathcal{F}_s$ is lognormally distributed and always stays positive. The expectation $\EQ[r_t | \mathcal{F}_s]$ and variance $\Var[r_t | \mathcal{F}_s]$
can be found analytically, \citep{BM2006}
\begin{align} \label{BKEQ}
\EQ[\bar{r}_t | \mathcal{F}_s] &= \exp\left[
e^{-k(t-s)} \log \bar{r}_s + k \int_s^t e^{-k(t-u)} \theta(u) du + \frac{1}{2}  \int_s^t \sigma^2(u) e^{-2 k(t-u)} du \right], \\
\Var[\bar{r}_t | \mathcal{F}_s] &= \exp\left[2 e^{-k(t-s)} \log \bar{r}_s + 2 k \int_s^t e^{-k(t-u)} \theta(u) du \right]
\left[ e^{2\mathcal{I}(s,t)} - e^{\mathcal{I}(s,t)}\right]. \nonumber
\end{align}
However, in the BK model, the price $F(t,\br)$ of a ZCB is not known in the closed form, since this model is not affine. Multiple good approximations have been developed in the literature using asymptotic expansions of various flavors; see, e.g., \citep{AntonovSpector2011,Capriotti2014,Horvath2017}, and also survey in \citep{Turfus2020}.

\section{An integral equation for the ZCB price in the BK model} \label{App2}

It is known that, written in terms of $z$, the corresponding PDE for the ZCB price $F(t,z)$  reads, \citep{andersen2010interest}
\begin{align} \label{VPDE}
0 &= \fp{F}{t} + \dfrac{1}{2}\sigma^2(t) \sop{F}{z} + \kappa(t) [\theta(t) - z] \fp{F}{z} - [s(t) + R e^z] V,
\end{align}
This equation should be solved subject to the terminal and boundary conditions, \citep{andersen2010interest} (see also discussion in \citep{CarrItkin2020jd})
\begin{equation} \label{termZCB1}
 F(T,z)  = 1, \qquad F(t,z)\Big|_{z \to \infty} = 0.
\end{equation}

Let us make the change of variables
\begin{align} \label{coefDefBK}
x &= a(t) z + b(t),  \qquad \tau = \frac{1}{2} \int_\nu^t \sigma (m)^2 e^{2 \int_\nu^m \kappa (m) \, dm} \, dm,
\qquad v(\tau, x) = e^{-\int_\nu^t s(t) \, dt}F(t,z), \\
a(t) &= e^{\int_\nu^t \kappa (m) \, dm}, \quad b(t) = -\int_\nu^t \theta (m) \kappa (m) e^{\int_\nu^m \kappa (m) \, dm} \, dm,. \nonumber
\end{align}
\noindent where $\nu = const$. As is discussed below, this constant can be chosen  to simplify the final expressions. With this change \eqref{VPDE} can be transformed to
\begin{align} \label{vEq}
0 &= \fp{v}{\tau}  + \sop{v}{x}  - \beta(t(\tau)) e^{\bar{a}(t(\tau)) x}v,  \\
\beta(t) &= \frac{2 R}{\sigma (t)^2} e^{-2 \int_\nu^t \kappa (m) \, dm -\frac{b(t)}{a(t)}}, \quad \bar{a}(t) = \frac{1}{a(t)} \nonumber \\
v(0,x) &= 1, \qquad \qquad \tau_0 = \frac{1}{2} \int_\nu^T \sigma (m)^2 e^{2 \int_\nu^m \kappa (m) \, dm} \, dm . \nonumber
\end{align}

Next, we apply the Fourier transform
\begin{equation*}
w(\omega, \tau) = \int_{-\infty}^{\infty} v(\tau,x) e^{\iu \omega x} dx,
\end{equation*}
\noindent to \eqref{vEq} to get the following problem
\begin{equation} \label{fEq}
w_\tau - \omega^2 w =  \beta(\tau) \int_{-\infty}^{\infty} v(\tau,x) e^{ \bar{a}(t) x} e^{\iu \omega x} dx,
\qquad w(\tau_0,\omega) = 2 \pi \delta (\omega ),
\end{equation}
\noindent where $\delta(\omega)$ is the Dirac delta function. The solution of \eqref{fEq} reads
\begin{equation}
w = e^{\omega^2 (\tau - \tau_0)} w(\tau_0, \omega) + \int_{\tau_0}^\tau e^{\omega^2 (\tau -k) } \beta(k) \int_{-\infty}^{\infty} F(k,\zeta) e^{ \bar{a}(k) \zeta -\int_{\nu(\tau)}^{t(\tau)} s(t) \, dt }  e^{\iu \omega \zeta} \,d\zeta \,dk.
\end{equation}
Applying the  inverse transform
\[ v(\tau,x) = \frac{1}{2\pi }  \int_{-\infty} ^{\infty} w(\tau,\omega) e^{-\iu \omega x} d\omega \]
\noindent yields
\begin{align}
v(\tau,x) &= 1 +\frac{1}{2\pi}\int_{-\infty}^{\infty}\int_{\tau_0}^\tau e^{\omega^2 (\tau -k) } \beta(k) \int_{-\infty}^{\infty} v(k,\zeta) e^{ \bar{a}(k) \zeta} e^{\iu \omega (\zeta-x)} d\zeta \, dk \, d\omega \nonumber \\
&= 1 -\frac{1}{2\sqrt{\pi}} \int_{-\infty}^{\infty}\int_{\tau_0}^\tau \frac{\beta(k) e^{\zeta \bar{a}(k)}}{\sqrt{\tau - k}} e^{-\frac{(x-\zeta)^2}{4(\tau -k)}} v(k,\zeta) d\zeta dk,
\end{align}
\noindent where in the last line the change of variables $k \to -k, \ \tau \to -\tau$ was made.

Thus, the ZCB price solves  the following two-dimensional Volterra integral equation of the second kind
\begin{equation} \label{intZCBbk}
v(\tau,x) = 1 - \frac{1}{2\sqrt{\pi}} \int_{-\infty}^{\infty}\int_{\tau_0}^\tau \frac{\beta(t(k)) e^{\bar{a}(t(k)) \zeta}}{\sqrt{\tau - k}} e^{-\frac{(x-\zeta)^2}{4(\tau -k)}} v(k,\zeta) d\zeta dk.
\end{equation}

It can also be obtained by applying the Duhamel's principle to \eqref{vEq}.

\section{Methods for solving \eqref{intZCBbk}} \label{RDTM}

\eqref{intZCBbk} is a two dimensional integral Volterra equation of the second kind.  Various authors have proposed efficient numerical methods for solving this type of equations. These methods include the block-by-block method, collocation and iterated collocation methods, the differential transform method (DTM), Galerkin and spectral Galerkin methods, multi-step collocation methods, and several other, see \citep{Torabi2019} and references therein. However, the complexity of the numerical methods (excluding the DTM) is at least $O(N^3)$, where $N$ is the number of computational nodes. On the other hand, $N$ could be taken relatively small compared, e.g., with the corresponding finite-difference (FD) method, if the high order quadrature rules are used when approximating the integrals.

Also, when applying all the methods mentioned above, the infinite interval should be replaced with a finite one. Another change of variables can do this, e.g., $\zeta \mapsto \tanh(\zeta)$. Then another class of methods can be used where the unknown function $v(\tau, x)$ is expanded into series on some basis. This basis could be a set of orthogonal functions, or Taylor series, etc.

However, a quick estimation of the solution of \eqref{intZCBbk} can be obtained along the lines of the reduced differential transform method (RDTM), \citep{Abazari2013}. The RDTM can be considered as an asymptotic solution of \eqref{intZCBbk} around some time $t = t_0$. It can be constructed with arbitrary precision. It is worth mentioning that the RDTM can not be directly applied to \eqref{intZCBbk} as the kernel in \eqref{intZCBbk} depends on $\tau$ itself. Therefore, we propose a modification of the RDTM suitable to handle this situation as well.

Next, we briefly present basic definitions of the RDTM and some theorems from \citep{Abazari2013} necessary to use this method for solving \eqref{intZCBbk}.

Consider a function of two variables $w(t,x)$, and suppose that it can be represented as a product of two single-variable functions $w(x,t) = f(x)g(t)$. The function $w(t,x)$ can be represented as
\begin{equation} \label{rdtm}
w(x,t) = \sum_{i=0}^\infty F(i)x^i \sum_{j=0}^\infty G(j) t^j = \sum_{i=0}^\infty \sum_{j=0}^\infty W_j(i)(i,j) x^i t^j,
\end{equation}
\noindent where $W(i,j) = F(i) G(j)$ is called the spectrum of $w(x,t)$.

To start with, we briefly present basic definitions of the RDTM and some theorems from \citep{Abazari2013} necessary to use this method for solving \eqref{intZCBbk}.

Consider a function of two variables $w(t,x)$, and suppose that it can be represented as a product of two single-variable functions $w(x,t) = f(x)g(t)$. The function $w(t,x)$ can be represented as.

If the double sum in \eqref{rdtm} is truncated to the $N$ terms in each variable, this expressions is the Poisson series of the input expression $w(x,t)$ with respect to the variables $(x,t)$ to order $N$ using the variable weights $W(i,j)$.

If $w(x,t)$ is an analytic function in the domain of interest, then the spectrum function
\begin{equation} \label{spec}
W_k(x) = \frac{1}{k!}\left[ \frac{\partial^k}{\partial t^k} w(x,t)\right]_{t=t_0}
\end{equation}
\noindent is called the reduced transformed function of $w(x,t)$. We use the notation where the lowercase $w(x,t)$
denotes the original function while the uppercase $W_k(x)$ stands for the reduced transformed function. The differential
inverse transform of $W_k(x)$ is defined as
\begin{equation} \label{invRDTM}
w(x,t) = \sum_{k=0}^\infty W_k(x) (t-t_0)^k.
\end{equation}
Combining \eqref{spec} and \eqref{invRDTM} one can get
\begin{equation} \label{relw}
w(x,t) = \sum_{k=0}^\infty \frac{1}{k!}\left[ \frac{\partial^k}{\partial t^k} w(x,t)\right]_{t=t_0} (t-t_0)^k.
\end{equation}

To proceed, we need the following fragment of Theorem 7 in \citep{Abazari2013}
\begin{theorem} \label{theo1}
Assume that $U_k(x), H_k(x)$ and $W_k(x)$ are the reduced differential transforms of the functions $u(x,t), h(x,t)$ and $w(x,t)$, respectively. If
\begin{equation} \label{prodRDTM}
w(x,t) = \int_{t_0}^t \int_{x_0}^x h(y,z) u(y,z) dy dz,
\end{equation}
\noindent then
\begin{equation} \label{iter1}
W_k(x) = \frac{1}{k} \int_{x_0}^x  \left( \sum_{\nu=0}^{k-1} H_\nu(y) U_{k-\nu-1}(y) \right) dy, \qquad k=1,2,\ldots.
\end{equation}
\end{theorem}
\begin{proof}
See \citep{Abazari2013}.
\end{proof}

Now one can observe that \eqref{intZCBbk} actually has the form of \eqref{prodRDTM} with $x = \infty, x_0 = -\infty, t_0 = 0$, and thus
\begin{align}
u(k, \zeta) &= -\beta(t(k)) e^{\zeta \bar{a}(t(k))} \frac{1}{2\sqrt{\pi (\tau-k)}} e^{-\frac{(x-\zeta)^2}{4(\tau-k)}}.
\end{align}

Our modification of the RDTM consists in eliminating the definition in \eqref{spec} for the function $u(k,\zeta)$.
Then, based on the definition of the reduced differential transform in \eqref{spec}, we get
\begin{align} \label{Udef}
U_0(s,\zeta) &= -\beta(t(s)) e^{\zeta \bar{a}(t(s))} \frac{1}{2\sqrt{\pi (\tau-s)}} e^{-\frac{(x-\zeta)^2}{4(\tau-s)}} , \\
U_1(s,\zeta) &= - \int_0^s \beta(t(k)) e^{\zeta \bar{a}(t(k))} \frac{1}{2\sqrt{\pi (s-k)}} e^{-\frac{(x-\zeta)^2}{4(s-k)}} dk, \nonumber
 \end{align}
\noindent and so on.

Let us denote the reduced differential transform of $v(\tau,x)$ as $W_k(x)$.  From \eqref{intZCBbk} it follows that $W_0(x) = 1$, as the double integral vanishes at $\tau=0$, and the following properties of the RDT hold
\begin{alignat*}{2}
w_k(x) &= u_k(x) \pm v_k(x) \quad &&\longmapsto \quad W_k(x) = U_k(x) \pm V_k(x),   \\
w(x,t) &= A = const, \quad &&\longmapsto \quad W_0(x) = 1, \quad W_k(x) = 0, \ k > 0,
\end{alignat*}
Then from \eqref{iter1}, and \eqref{Udef} we have\footnote{This expression now contains an integral in time since we eliminated the Taylor series expansion in \eqref{spec}.}
\begin{equation} \label{w1}
W_1(x) =  \int_0^\tau \int_{-\infty}^\infty  W_0(s, \zeta) U_0(s,\zeta) d\zeta ds =
\int_0^\tau e^{\ba(s) (\ba(s) (\tau -s)+x)} ds.
\end{equation}
The next iteration reads
\begin{align} \label{w2}
W_2(x) &= \int_0^\tau \int_{-\infty}^\infty  \left[ W_1(\zeta) U_0(\zeta) + W_0(\zeta)U_1(\zeta) \right]d \zeta \, d s = I_1 + I2, \\
I_1 &= \int_0^\tau  \int_0^s  \beta (k) \beta (s) e^{a(k) (2 a(s) (\tau -s)+x)+a(k)^2 (\tau -k)+a(s) (a(s) (\tau -s)+x)} d k \, d s, \nonumber \\
I_2 &= - \int_0^\tau  \int_0^s \beta(t(k)) e^{\zeta \bar{a}(t(k))} \frac{1}{2\sqrt{\pi (s-k)}} e^{-\frac{(x-\zeta)^2}{4(s-k)}}  d k \, d s, \nonumber
\end{align}
\noindent etc.

Once all the terms $W_k(x), \ k=1,\ldots,N$ are found, the final representation of the solution follows from the inverse formula \eqref{invRDTM} changed according to our modification of the RDTM
\begin{equation} \label{finRDTM}
v(\tau,x) = \sum_{k=0}^N W_k(x).
\end{equation}
The time-integrals in \eqref{w1}, \eqref{w2} can be computed either numerically, or analytically if functions $\ba(\tau), \beta(\tau)$ could be expanded into series on $\tau$ around some $\tau_0$. In the latter case the method becomes almost identical to the original RDTM. When doing so, one has to remember that derivatives of $a(\tau), \beta(\tau)$ are the derivatives on $\tau$ while the definitions of these functions in \eqref{coefDefBK} are given in terms of $t = t(\tau)$. The latter map is also given in \eqref{coefDefBK}.

\subsection{Numerical example}

To test the RDTM as applied to our problem, we solve \eqref{intZCBbk} by using the modified RDTM described in Section~\ref{RDTM}. Here we use the following explicit form for $\kappa(t), \tth(t), \sigma(t)$
\begin{equation} \label{ex}
\kappa(t) = \kappa_0, \qquad \tth(t) = \tth_0 e^{\theta_1 t}, \qquad \sigma(t) = \sigma_0 e^{-\sigma_1 t},
\end{equation}
\noindent where $\kappa_0, \theta_0, \sigma_0, \theta_1, \sigma_1$ are constants. We also assume $s(t) = 0$, and $R = 1$. With these definitions, we have
\begin{align} \label{exCoef}
a(t) &= e^{-k (t-\nu)}, \quad b(t) = \left( e^{\theta_1 \nu} -
e^{t(\kappa_0  +\theta_1) - \kappa_0 \nu} \right) \frac{\kappa_0 \theta_0}{\kappa_0 + \theta_0}, \\
\beta(t) &= \frac{2}{\sigma^2 _0} \exp\left[
\frac{2 \theta_0}{\kappa_0 + \theta_1} \left( e^{\theta_1 t} - e^{\nu (\theta_1 + \kappa_0)  - \kappa_0 t } \right)
+2  t(\sigma_1 - \kappa_0) + 2 \kappa_0 \nu \right], \nonumber \\
t(\tau) &=  \frac{\log \left(e^{-2 \sigma_1 \nu} - \frac{2 \tau  (\kappa_0 - 2 \sigma_1)}{\sigma^2_0}\right) + \kappa_0  \nu}{\kappa_0 - 2\sigma_1}. \nonumber
\end{align}

\begin{table}[!htb]
\begin{center}
\caption{Parameters of the test.}
\label{tab1}
\begin{tabular}{|c|c|c|c|c|c|}
\hline
$r_0$ & $\kappa_0$ & $\theta_0$ & $\sigma_0$ & $\theta_1$ & $\sigma_1$ \\
\hline
0.01 & 1.0 & 0.05 & 0.5 & 0.2 & 0.2   \\
\hline
\end{tabular}
\end{center}
\end{table}

\begin{table}[htbp]
  \centering
  \caption{Prices of ZCB bonds with different maturities computed by using the modified RDTM with 1 and 2 terms, and the FD difference method.}
\label{tab2}%
\small
\begin{tabular}{|l|r|r|r|r|r|r|}
\cmidrule{2-7}    \multicolumn{1}{r|}{} & \multicolumn{6}{c|}{\textbf{ZCB price, \$}} \\
    \midrule
    \textbf{T} & \textbf{0.0833} & \textbf{0.3} & \textbf{0.5} & \textbf{1.0} & \textbf{2.0} & \textbf{5.0} \\
    \hline
    \textbf{FD} & 0.9990 & 0.9938 & 0.9839 & 0.9216 & 0.6201 & 0.0444 \\
    \hline
    \textbf{RDTM-1} & 0.9990 & 0.9950 & 0.9887 & 0.9617 & 0.5857 & 0.2477 \\
    \hline
    \textbf{RDTM-2} & 0.9990 & 0.9949 & 0.9883 & 0.9587 & 0.6186 & 0.6087 \\
    \hline
    \multicolumn{1}{r|}{} & \multicolumn{6}{c|}{\textbf{Difference with the FD method, \%}} \\
    \hline
    \textbf{T} & \textbf{0.0833} & \textbf{0.3} & \textbf{0.5} & \textbf{1.0} & \textbf{2.0} & \textbf{5.0} \\
    \hline
    \textbf{FD} & 0.0000 & 0.0000 & 0.0000 & 0.0000 & 0.0000 & 0.0000 \\
    \hline
    \textbf{RDTM-1} & 0.0047 & 0.1142 & 0.4907 & 4.3563 & -5.5489 & 457.6831 \\
    \hline
    \textbf{RDTM-2} & 0.0042 & 0.1039 & 0.4494 & 4.0263 & -0.2491 & 1270.3588 \\
    \bottomrule
    \end{tabular}%
\end{table}%

We present the model parameters for this test in Table~\ref{tab1}. We run the test for a set of maturities $T \in [1/12, 0.3,0.5,1,2,5]$. We show the corresponding ZCB prices in Table~\ref{tab2}.
As a benchmark, we use the solution of \eqref{VPDE} obtained by using the FD method described above. The modified RDTM provides reasonable accuracy for small maturities (up to 2 years), while for $T > 2$, one or two terms in the expansion are insufficient to get the correct price. Therefore, for larger $T$, the \eqref{intZCBbk} has to either be solved numerically or more terms should be taken in the RDTM.

Note, that there are at least two choices for $\nu$ in \eqref{coefDefBK}: $\nu=0$ and $\nu=T$. We found that for small $T$ the choice $\nu=0$ provides slightly better results, while for $T > 1$ it is better to use $\nu=T$.

Also, note that this method, in some sense, is similar to that in \citep{Capriotti2014}. However, we solve an integral equation instead of a PDE in \citep{Capriotti2014}. Besides, there is a difference in parametrization, since we assume that all the parameters are time-dependent. In contrast, in \citep{Capriotti2014}, all model parameters are constant.

As far as the performance of the RDTM is concerned, we compared it with the performance of the FD method applied to the forward equation (the forward analog of \eqref{VPDE}. The elapsed time of getting the ZCB prices by solving such the equation is 40 msec while using the RDTM even with the numerical computation of all integrals in \eqref{w2} takes 13 msec. Therefore, this method allows fast calculation of ZCB prices in the time-dependent BK model for $T < 2$.

\end{document}